%% file: main.tex
\documentclass[final,twocolumn,times]{elsarticle}

\usepackage{amssymb}

\usepackage{makecell}
\usepackage{enumitem}
 \usepackage{amsthm, amsmath}
 \usepackage{dsfont}
 \newtheorem{thm}{Theorem}
\newtheorem{lemma}[thm]{Lemma}
\newtheorem{observation}[thm]{Observation}
\newdefinition{rmk}{Remark}
\newproof{pf}{Proof}
\newproof{pot}{Proof of Theorem \ref{thm2}}
\newtheorem{reductionrule}{\textit{Reduction rule}}
\newcommand{\defproblem}[3]{
	\vspace{1mm}
	\noindent\fbox{
		\begin{minipage}{0.47\textwidth}
			\begin{tabular*}{\textwidth}{@{\extracolsep{\fill}}lr} #1 \\ \end{tabular*}
			{\bf{Input:}} #2  \\
			{\bf{Question:}} #3
		\end{minipage}
	}
	\vspace{1mm}
}

\newtheorem{proposition}[thm]{\textit{Proposition}}

\newtheorem{corollary}[thm]{\textit{Corollary}}

\newcommand{\VC}{\textsc{Vertex Cover} }
\newcommand{\RBSC}{\textsc{Red Blue Set Cover }}
\newcommand{\WRBSC}{\textsc{Weighted Red Blue Set Cover}}
\newcommand{\SC}{\textsc{Set Cover }}

\newcommand{\AHG}{Axis-parallel hyperplane }

\newcommand{\AQG}{Axis-parallel Quadrants }
\newcommand{\ASG}{Axis-parallel  Skylines }

\newcommand{\AH}{\textsc{AxRBSC-hyperplanes }}
\newcommand{\AL}{\textsc{AxRBSC-lines}}
\newcommand{\AQ}{\textsc{AxRBSC-quadrants }}
\newcommand{\AS}{\textsc{AxRBSC-skylines }}
\newcommand{\AST}{\textsc{AxRBSC-bi directional strips }}
\newcommand{\dHS}{\textsc{d-Hitting Set }}

\newcommand{\SQ}{\textsc{Square Stabbing }}
\newcommand{\RS}{\textsc{Rectangle Stabbing} }
\newcommand{\calH}{\mathcal{H}}

\journal{Information Processing Letters}

\begin{document}
	
	\begin{frontmatter}
	
		
		
\title{Red Blue Set Cover Problem on Axis-Parallel Hyperplanes and Other Objects}

	\author{V P Abidha\fnref{fn1} }
	\ead{abidha.vp@iiitb.ac.in}
	\fntext[fn1]{First author is supported by the TCS Research scholar program}	

		\author{Pradeesha Ashok}
		\ead{pradeesha@iiitb.ac.in}
		\address{International Institute of Information Technology Bangalore, India}
		

		\begin{abstract} 
Given a universe $\mathcal{U}=R \cup B$ of a finite set of red elements $R$, and a finite set of blue elements $B$ and a family $\mathcal{F}$ of subsets of $\mathcal{U}$, the \RBSC problem is to find a subset $\mathcal{F}'$ of $\mathcal{F}$ that covers all blue elements of $B$  and minimum number of red elements from $R$.


We prove that the \RBSC problem is NP-hard even when $R$ and $B$ respectively are sets of red and blue points in ${\rm I\!R}^2$ and 
the sets in $\mathcal{F}$ are defined by axis-parallel lines i.e, every set is a maximal set of points with the same $x$ or $y$ coordinate.
We then study the parameterized complexity of a generalization of this problem,
 where $\mathcal{U}$ is a set of points in ${\rm I\!R}^d$ and $\mathcal{F}$ is a collection of set of axis-parallel hyperplanes in ${\rm I\!R}^d$, under different parameterizations. For every parameter, we show that the problem is fixed-parameter tractable and also show the existence of a polynomial kernel. 
 
\noindent  We further consider the \RBSC problem for some special types of rectangles in ${\rm I\!R}^2$.

\end{abstract}
		


		\begin{keyword}
			
			
Red blue set cover \sep Axis parallel hyperplane \sep Parameterized complexity \sep Polynomial kernels \sep Axis parallel quadrants.		
		
	\end{keyword}
		
	\end{frontmatter}

	

%
%
%

\section{Introduction}
\input{intro.tex}
\section{Preliminaries}
\input{preliminaries.tex}

\section{NP-completeness}
\label{NPC}

\input{NPComplete.tex}

\section{RBSC on Axis-parallel hyperplanes}
\label{hyper}
\noindent From Corollary \ref{Cor1}, we know that \AH is NP-complete. In this section, we show fixed parameter tractability and the existence of polynomial kernels for the parameterized \AH problem, under different parameterizations.
\subsection{Parameterizing by $k_r$}
\label{red}
\input{Para_K_r.tex}

\subsection{Parameterizing by $b$}
\label{bl}
\input{Para_blue.tex}

\subsection{Parameterizing by $h$}
\label{hyp}
\input{Para_line.tex}

%
%
%
%

%
%
%
%

\section{Red Blue Set Cover on Skylines}
\label{Sk}
\input{reduction.tex}


\section{Quadrants}
\label{Qdr}
\input{Q1.tex}



	
	
	
\end{document}

%% file: intro.tex
The \SC problem and its variants are well studied problems in Computer Science.
Given a universe $\mathcal{U}$ with $n$ elements and a family $\mathcal{F}$ of subsets of $\mathcal{U}$, the \SC problem is to find a subset $\mathcal{F}' \subseteq \mathcal{F}$ that covers all elements in $\mathcal{U}$ and $|\mathcal{F}'|$ is minimized. Here $u \in \mathcal{U}$ is said to be covered by $\mathcal{F}'$, if $\exists F \in \mathcal{F}'$ such that $u \in F$.

Researchers
have looked into many variations of this question. We study the \RBSC problem, which is a generalization of the \SC problem, where $\mathcal{U}$ is bi-chromatic.
Given a universe $\mathcal{U}=R \cup B$ of a finite set $R$ of red elements, and a finite set  $B$ of blue elements and a family $\mathcal{F}$ of subsets of $\mathcal{U}$, the \RBSC problem is to find a subset $\mathcal{F}'$ of $\mathcal{F}$ that covers all blue elements of $B$ and the minimum number of red elements from $R$. Note that here the number of sets in $\mathcal{F}'$ is not optimized. The \SC problem can be considered as a special case of the \RBSC problem where every set contains a distinct red element. This implies that the \RBSC problem is as hard as the classical \SC problem.


The \RBSC problem was introduced by Carr et al.~\cite{RBSCintro} in 2000, motivated by applications in data mining. Several algorithmic results for the \RBSC problem are given in \cite{LabelMax,C1P,DomC1P,positive_negative}. A special case of the \RBSC problem is the \textsc{Geometric \RBSC} problem, where $\mathcal{U}= R \cup B$ and a family $\mathcal{F}$ is defined by a collection of geometric objects, which contains points from  $\mathcal{U}$.
%

In the \textsc{Geometric \RBSC} problem, a wide variety of geometrical objects are studied.
Chan and Hu \cite{Chan2013GeometricRS} proved that the problem is NP-hard even when the objects are unit squares in $\mathbb{R}^2$, and gave the first polynomial-time approximation scheme (PTAS) for this problem. Pandit and Mudgal \cite{UnitDisk} have shown a constant factor approximation algorithm for the \RBSC problem on unit disks. A study of different axis-parallel systems has been conducted by RR Madireddy et al.\cite{RBSC_AX}. They showed the APX-hardness on axis-parallel set systems like axis-parallel rectangle, axis-parallel strips and axis-parallel line segments. The \textsc{Geometric} \RBSC problem is also studied in the context of parameterized complexity. 
Ashok et al.\cite{RBSC} investigates the parameterized complexity of the \RBSC problem on general lines in $\mathbb{R}^2$ under an array of parameters. 
In this paper, we mainly study the parameterized complexity of \RBSC problem on axis-parallel lines in $\mathbb{R}^2$ and axis-parallel hyperplanes in $\mathbb{R}^d$. This can be seen as a continuation of the results given in \cite{RBSC}, where lines and hyperplanes of arbitrary orientation are considered. We show the existence of positive results when the objects are restricted to be axis-parallel.
\vspace{.1cm}




\noindent\textbf{Problems studied and results :}

 \vspace{0.2cm}
\noindent We study the following decision variant of the \RBSC problem. 

 \vspace{0.1cm}
\defproblem{$\textsc{\RBSC}(\textsc{RBSC})$}{A universe $\mathcal{U}=R \cup B$, a family  $\mathcal{F} \subseteq 2^{\mathcal{U}}$, an integer $k_r$.}{Is there a subset $\mathcal{F'}\subseteq \mathcal{F}$  that covers all points in $B$ and at most $k_r$ points from $R$ ?}  


\noindent We consider the \RBSC problem for the following objects.
\begin{enumerate}
\item \AL~: We study the \RBSC problem when $\mathcal{U}$ is a set of points in $\mathbb{R}^2$ and $\mathcal{F}$ is defined by a set of axis-parallel lines(refered to as \AL). We show that \AL~is NP-complete. See Section \ref{NPC}.
\item \AH : We study a generalization of \AL~to $d$ dimensions. We consider the problem where $\mathcal{U}$ is a set of points in $\mathbb{R}^d$ and $\mathcal{F}$ is defined by a set of axis-parallel hyperplanes. This problem is NP-complete since \AL~is NP-complete. We study the parameterized complexity of \AH under different parameterizations. For each parameter, we consider whether the problem is fixed parameter tractable and if yes, we investigate whether the problem admits polynomial kernels. 
See Section \ref{hyper}. A comparison of results for arbiitrary hyperplanes and axis-parallel hyperplanes is given in Table \ref{TT1}.
\end{enumerate}
\noindent We further consider some special types of rectangles as follows.
\begin{enumerate}[resume]
\item \AS : Bi-directional strips refer to axis-parallel rectangles which are unbounded on two opposite sides. They are either of the form $[a,b]\times (-\infty,\infty)$ or $ (-\infty,\infty) \times [a,b]$. Note that axis-parallel lines is a special case of bi-directional strips and hence the \RBSC problem on bi-directional strips is NP-complete. We study the \RBSC problem on a generalization of strips called skylines (referred to as \AS) and show that this problem is W-hard and therefore unlikely to admit FPT algorithms. See Section \ref{Sk}.
\item \AQ : We further consider quadrants, a special case of skylines and show that the \RBSC problem on quadrants is polynomial time solvable. See Section \ref{Qdr}.
\end{enumerate}
%


\begin{table}[h!]
	\centering
	\begin{tabular}{|c|c|c|}
		\hline
		\textbf{Parameter} & \textbf{Axis-parallel}  & \textbf{Arbitrary }   \\ 
		\hline
		$k_r$ & FPT and poly.kernel & W[2]-hard\\  
		\hline
		$b$& FPT and poly.kernel & W[1]-hard\\
		\hline
		$l$& FPT and poly.kernel & FPT and poly.kernel\\  
		\hline
		$r$& FPT and poly.kernel & FPT and no poly.kernel\\  
		\hline 
			\end{tabular}
	\vspace{0.03cm}
	\caption{Comparison of results for arbitrary hyperplanes(taken from \cite{RBSC}) and axis-parallel hyperplanes.}
	\label{TT1}
\end{table} 

%% file: preliminaries.tex
In this section, we give definitions and results that will be used in subsequent sections.

\noindent The formal definitions of objects that we consider are as follows : \\
\noindent\textbf{\AHG}: A hyperplane $H$ is a subspace whose dimension is one less than that of its ambient space. Here we consider hyperplanes parallel to any one of the axis in ${\rm I\!R}^d$ .\\
\noindent\textbf{\ASG}: Skylines are axis-parallel rectangles that are unbounded in one direction. Every skyline is bounded by three axis-parallel edges. We define a horizontal skyline as a rectangle of the form $(-\infty,x]\times [y_1,y_2]$ and a vertical skyline as a rectangle of the form  $[x_1,x_2]\times(-\infty,y]$ in ${\rm I\!R}^2$. \\
\noindent\textbf{\AQG}: Quadrants are axis-parallel rectangles that are bounded by two adjacent axis-parallel edges. We define \emph{quadrants of the first type}, as rectangles of the form $[x,+\infty)\times [y,+\infty)$, \emph{quadrants of the second type} as rectangles of the form  $(-\infty,x]\times [y,+\infty),$  \emph{quadrants of the third type} as rectangles of the form  $(-\infty,x]\times [y,-\infty)$ and \emph{quadrants of the fourth type} as rectangles of the form  $[x,+\infty)\times [y,-\infty)$. In this paper, we consider quadrants of the first type in ${\rm I\!R}^2$. Note that the results on  quadrants of the first type can be extended to quadrants of other types too.

\vspace{.15cm}
In this paper, we study the parameterized complexity of the problems using different parameters.

\noindent\textbf{Parameterized complexity \cite{PACygan} : }Parameterized complexity was introduced as a means to tackle NP-hardness. A parameterized problem instance is a pair $(\Pi,k)$ where $\Pi$ is the actual problem instance and $k$ is called a parameter. $(\Pi,k)$ is said to be \emph{fixed parameter tractable (FPT)} if there exists an algorithm that solves $\Pi$ in $O(f(k)|\Pi|^{O(1)})$ time, where $f(k)$ is a computable function independent of $|\Pi|$. 

Parameterized complexity is interesting when a problem is known to be NP hard. Here we aim to restrict the combinatorial explosion to a parameter, where the parameter is much smaller than the size of input.

\vspace{.05cm}
\noindent\textbf{W-Hierarchy :} Theory of intractability of parameterized problems orders the problems into a hierarchy called the W-hierarchy based on its complexity. It is organized as FPT $\subseteq W[1] \subseteq W[2] \cdots $. Under standard complexity theoretical assumptions, a problem which is $W[i]$-hard does not admit FPT algorithms, for $i >0$. To show intractability there exists a concept of parameterized reduction. 
%

This technique can be used to transforms an instance of a parameterized problem $\Pi_A$ into an instance of a parameterized problem $\Pi_B$, mapping yes-instances for $\Pi_A$ to yes-instances for $\Pi_B$ and vice versa. If this transformation can be done in fixed-parameter tractable time, this implies that if $\Pi_B$ is fixed-parameter tractable, then so is $\Pi_A$; equivalently, if $\Pi_A$ is not fixed-parameter tractable, then neither is $\Pi_B$. This reduction is similar to that of polynomial reduction in classical complexity. Parameterized reductions are a powerful technique for relating problems to each other.

Another important concept related to parameterized complexity is kernelization.

\vspace{.05cm}
\noindent\textbf{Kernelization :}
A kernelization algorithm($A$) is usually a set of reduction rules that, given an instance $(I,k)$ of a parameterized problem $\Pi$, runs in polynomial time and returns an equivalent instance $(I',k')$ of $\Pi$. Moreover, we require that $|I'|+|k'| \leq g(k)$ for some computable function $g$ depending on $k$. 


The output instance $I'$ is called the kernel, and the function $g$ is referred to as the size of the kernel. If $g(k)=k^{O(1)}$, then we say that $\Pi$ admits a polynomial kernel.  The set of reduction rules used in kernelization is said to be \textit{safe}, if it converts a problem instance to an equivalent instance.

One important result in parameterized complexity is that a parameterized problem $\Pi$ is FPT if and only if it admits a kernel. Thus for a problem which is in FPT, an interesting question is to know whether there exists a polynomial kernel or not.

Another technique used in  designing of parameterized algorithms is Bounded Search Trees.

\vspace{.05cm}
\noindent\textbf{Bounded Search Trees:} In this technique, we try to build a feasible solution to the parameterized problem by making a sequence of decisions.
Nodes of a search tree correspond to problem instances, and children of a node correspond to a number of decisions that can be taken during the search. 
While designing an FPT algorithm, if we unravel the recursion tree, the depth of this tree  and branching factor will be bounded by a function of  $k$. 
This will bound the total size of the tree as a function of $k$ 
and the resulting algorithm is a fixed parameter tractable algorithm.



\noindent Now we give some more definitions and results.

\vspace{.05cm}
\noindent\textbf{Orthogonal Convex Hull}: A simple polygon $C$ is said to be \emph{orthogonal convex} if every vertical or horizontal line intersects the boundary of $C$ at zero, one or two points. Orthogonal convex hull of a point set $P$ is defined as the smallest area orthogonal convex polygon that contains $P$. 
An algorithm that finds the\textit{ orthogonal convex hull} of a set of $n$ points in $O(n \log n)$ time is presented in \cite{orthoConvex}.



\noindent We define the \textsc{Red Blue Set Cover} problem studied in weighted settings as follows.

\defproblem{$\textsc{\WRBSC}$}{A bi-chromatic point set $\mathcal{U}$ of blue and red elements and a family $\mathcal{F}$ of subsets of $\mathcal{U}$, a weight function $w : R\rightarrow \mathbb{N}$, an integer $k$.}{Is there a subfamily of $\mathcal{F}$ that covers all blue points and total weight of red points covered is less than or equal to $k$?}

\vspace{.1cm}
\noindent We define the \textsc{\dHS} problem as follows.

\vspace{.1cm}
\defproblem{$\textsc{\dHS}$}{A collection $\mathcal{F}$ of subsets of size $d$ of a finite set $S$, and a positive integer $k$.}{Does there exist a subset $S' \subseteq S$ with $|S'| \leq k$ such that $S'$ contains at least one element from each subset in $\mathcal{F}$ ?}

\begin{proposition}
	\label{DH}
 \cite{3hittingset} The \dHS problem can be solved in $O({c_d}^k + n)$ time, where  $c_d = \frac{d-1+\sqrt{(d-1)^2+4}}{2} $.
\end{proposition}
%
%
\noindent\textbf{Notations:} 
$P=R \cup B$ is a point set, where $R$ is a set of $r$ red points and $B$ is a set of $b$ blue points and $|R|+|B|=n$. 
For any point $p$, let $x(p)$ and $y(p)$ denote the $x$ and $y$ coordinate values of $p$ respectively. Let $[d]$ be the set, $\{1, 2, \cdots, d\}$ throughout the paper.
Also we use the notation $O^{*}$ that hides the polynomial factor in the running time.

\vspace{0.2cm}
\noindent We use following reduction rules exhaustively on any instance of the \RBSC problem.
\begin{reductionrule}
	\label{Reduction rule 1}
	If a blue element $b$ is contained in only one set $S \in \mathcal{F}'$,
where $\mathcal{F}' \subset \mathcal{F}$ is a solution that covers all blue points and $k_r$ red points, then delete $S$ from $\mathcal{F}'$, $S \cap B$ from $B$ and $S \cap R$ from $R$, set $k_r= k_r-|S \cap R|$.
\end{reductionrule}
\begin{reductionrule}
	\label{Reduction rule 2}
	If a set $S$ contains only red points, then delete $S$.
\end{reductionrule}
\begin{reductionrule}
	\label{Reduction rule 3}
	If a set $S$ contains only blue points, then delete $S$ from $\mathcal{F}$ and $S \cap B$ from $B$. 
\end{reductionrule}
\begin{reductionrule}
	\label{Reduction rule 4}
	If a set $S$ contains more than $k_r$ red points, then delete $S$ from $\mathcal{F}$. 
\end{reductionrule}


\begin{lemma}
	\label{RRProof}
	Reduction rules \ref{Reduction rule 1},\ref{Reduction rule 2} ,\ref{Reduction rule 3} and \ref{Reduction rule 4} are safe.
\end{lemma}

\begin{proof}
	The safety of Reduction rule \ref{Reduction rule 1} follows from the fact that if $S$ is not part of the solution family, then there is no other set to cover $b$.  
	If $S \in \mathcal{F}$ contains only red points and $\mathcal{F}' \subset \mathcal{F}$ is a solution that covers all blue points and $k_r$ red points such that $S \in \mathcal{F}'$, then $\mathcal{F}'\backslash \{S\}$ covers all blue points and at most $k_r$ red points. This shows that Reduction rule \ref{Reduction rule 2} is safe. Similarly, if $S$ contains only blue points and $\mathcal{F}'$ is a solution that covers all blue points and $k_r$ red points then $\mathcal{F}' \cup \{S\}$  is also a solution that covers all blue points and $k_r$ red points. Thus, Reduction rule \ref{Reduction rule 3} is safe. The safety of Reduction rule \ref{Reduction rule 4} follows from the fact that the budget is at most $k_r$ and if a set contains more than $k_r$ red points, then it can not be a part of the solution.
\end{proof}

%% file: NPComplete.tex
%
%
\begin{thm}
	\AL~is NP-complete.
\end{thm}
\begin{proof}
	We give a reduction from the \VC{} problem. Let $G=(V, E)$ be an instance of the \VC problem, where $|V|=n$ and $|E|=m$. We give an instance of \AL, $(R \cup B, \mathcal{L})$, corresponding to $G$. For every $v_i \in V$, we add a red point $r_i$ to $R$ with coordinate $(i, i) \in {\rm I\!R}^2$. Add horizontal and vertical lines that pass through $r_i$, for $1 \leq i \leq n$, to $\mathcal{L}$. For every edge $e(v_i,v_j) \in E$ (assume $i< j$), we add a blue point $b_{ij}$ to $B$ at the intersection of the horizontal line that passes through $r_i$ and the vertical line that passes through $r_j$.

	We now claim that $G=(V, E)$ admits a vertex cover of size $k$ if and only if $(R \cup B, \mathcal{L})$ admits a solution of size $k$. Let $V' \subset V$ be a vertex cover of $G$ of size $k$. Without loss of generality, let $V' = \{v_1, v_2, . . , v_k\}$. Now the set of lines in $\mathcal{L}$ that pass through the points $r_i$, $1 \leq i \leq k$, covers all the blue points and $k$ red points.
	
	
	Conversely, suppose $(R \cup B, \mathcal{L})$  has a solution that covers at most $k$ red points,
	say, $R' = \{r_1, r_2, . . , r_k\}$. Then the vertices of $V$ corresponding to red points in $R'$ is a vertex cover of $G$.  
	
\AL~is in NP since for a given solution of \AL, it easy to verify whether the solution covers all blue points and covers at most $k$ red points in polynomial time.	
	
	
\noindent Now the result follows from the fact that \VC is NP-complete~\cite{VC}.
\end{proof}

\begin{corollary}
	\label{Cor1}
	\AH  is NP complete.
\end{corollary}

\begin{corollary}
	 \AST is NP complete. 
\end{corollary}
%

%% file: Para_K_r.tex
In this section, we consider the \AH problem parameterized by $k_r$, where $k_r$ is the number of red points in the solution. Let $(R \cup B,\calH,k_r)$ be an instance of the \AH problem parameterized by the number of red points $k_r$, where $R \cup B$ and $\calH$ are set of points and set of hyperplanes in the instance respectively.

We show the \AH problem is FPT parameterized by $k_r$ using the bounded search tree technique.
We apply reduction rules \ref{Reduction rule 1}, \ref{Reduction rule 2}, \ref{Reduction rule 3} and \ref{Reduction rule 4} on $(R \cup B, \mathcal{H},k_r)$   exhaustively. Now every hyperplane contains at least one red point and at least one blue point. Every blue point is covered by at most $d$ hyperplanes.  Thus, for a blue point, we check with the set of hyperplanes that passes through it in a bruteforce manner and  we can recursively solve the remaining problem. Since every hyper plane contains at least one red point, the associated budget $k_r$ is reduced by at least one in the subproblem.
The bounded search tree algorithm returns YES, when $B=\phi$ else if $k_r=0$ and $B \neq \phi$  it returns NO.

Thus giving us an algorithm that runs in $O^*(d^{k_r})$ time. The next result improves this running time.

\begin{thm}
	\label{FPT_red} 
	The problem  \AH can be solved in time $O({c_d}^{k_r}.n^{O(1)})$ where $c_d = \frac{d-1+\sqrt{(d-1)^2+4}}{2} $.
\end{thm}

\begin{proof}
	 

Our  bounded search tree algorithm considers a blue point $b$ such that at least one hyperplane that contains $b$ has two red points and recursively solve the remaining problem on the hyperplane that covers $b$ in the solution. Thus, $k_r$ drops by at least $2$ in at least one branch and drops by at least one in all other branches. This gives us a branching  algorithm that runs in $O({c_{d}}^{k_r}.n^{O(1)})$ time where $c_d = \frac{d-1+\sqrt{(d-1)^2+4}}{2}$.
If the instance does not contain such a blue point (i.e, every hyperplane contains exactly one red point), then the instance can be reduced to an instance of the \dHS problem parameterized by solution size. Let $(R\cup B, \calH, k_r)$ be an instance of \AH such that every $H \in \calH$ contains exactly one red point from $R$. We construct an instance $(U, \mathcal{F},k)$ of \dHS with $U = R,  k = k_r$ and $\mathcal{F} $ contains sets corresponding to every blue point, where a set corresponding to a blue point $b$ contains all the red points contained in  the at most $d$ hyperplanes that contain $b$. It is easy to see that $(R\cup B, \calH, k_r)$ is a YES instance if and only if $(U, \mathcal{F},k)$ is a YES instance. By Proposition \ref{DH},  $(U, \mathcal{F},k)$ can be solved using a branching algorithm that runs in $O({c_{d}}^{k}.n^{O(1)})$ time.
	
%
%
%
\end{proof}

%
\noindent \textbf{Kernelization}
\vspace{.1cm}
\noindent We apply the following set of reduction rules for all integers $\delta = 1$ to $ d-1$ exhaustively. For $2\leq \delta \leq d-1$, we apply the reduction rule corresponding to $\delta$ only when the reduction rule corresponding to  $(\delta-1)$ is not applicable. For a point $p \in \mathbb{R}^d$, $x_i(p)$ denotes the value of the $i^{th}$ coordinate of $p$, for $1 \leq i \leq d$.


\begin{reductionrule}[$\delta $]
\label{SameCoordPoints}
Repeat for all $D' \subset [d]$ such that $|D'| = d-\delta$. 
\\Let $B' \subseteq B$ be such that for all $p,q \in B'$, $x_i(p) = x_i(q)$ for all $i \in D'$. If $|B'| > \delta!(k_r)^{\delta} +1$ then delete all but $\delta! k_r^{\delta}+1$ points of $B'$.
\end{reductionrule}
\begin{lemma}
Reduction Rule \ref{SameCoordPoints} ($\delta $) is safe.
\end{lemma}

\begin{proof}
The proof follows from the fact that in any set of hyperplanes $\mathcal{H'}$ that covers at most $k_r$ red points, all the points of $B'$ are covered by the same hyperplane. Assume, for contradiction, that no hyperplane in $\mathcal{H'}$ contains all points in $B'$. Thus $B'$ is covered by a set of hyperplanes, say $\mathcal{H}_1 \subseteq \mathcal{H'} $ such that no  hyperplane in $\mathcal{H}_1$ is parallel to the $i^{th}$ co-ordinate axis, for all $i \in D'$. Without loss of generality, assume that $D' = \{1,2,\dots,d-\delta\}$. We know that every hyperplane contains at least one red point (by Reduction Rule~\ref{Reduction rule 2}) and parallel hyperplanes cannot share a red point. Therefore, $\mathcal{H}_1$ can have at most $k_r$ hyperplanes that are parallel to the $j^{th}$ co-ordinate axis, for all $d-\delta+1 \leq j \leq d$. Thus $\mathcal{H}_1$ contains at most $k_r.\delta$ hyperplanes. 

Let $H$ be a hyperplane in $\mathcal{H}_1$ that is parallel to the $j^{th}$ co-ordinate axis for $j \notin D'$. Then for any two points $p,q \in H \cap B'$, $x_j(p)=x_j(q)$ and $x_i(p) =x_i(q)$ for $1\leq i \leq d-\delta$.  Therefore, hyperplane in $\mathcal{H}_1$ can contain at most $(\delta-1)!(k_r)^{\delta -1}$ points from $B'$.  Otherwise, Reduction Rule \ref{SameCoordPoints} $(\delta-1)$ can be applied, which is a contradiction. 

Thus $\mathcal{H}_1$ covers at most $\delta!(k_r)^{\delta }$ points. This contradicts our assumption that $\mathcal{H}_1$ covers all points of $B'$. 
\end{proof}

\begin{thm}\label{k_kernal}
\AH admits a kernelization algorithm that returns $(R \cup B, \calH)$ with $|B| \leq d! k_r^d$, $|R| \leq d.d!.k_r^{d+1}$, $|\calH| \leq d.d! k_r^d$.
\end{thm}

\begin{proof}	
	Let $(R \cup B, \mathcal{H},k_r)$ be an instance of \AH such that none of the reduction rules are applicable. Then, if it is a YES instance, all the points in $B$ can be covered by at most $k_r.d$ hyperplanes. Any hyperplane can contain at most $(d-1)!k_r^{(d-1)}$ blue points, otherwise Reduction Rule \ref{SameCoordPoints} $(d-1)$ can be applied. Thus a YES instance can have at most $d! k_r^d$ blue points. 
	
	Every hyperplane in $\calH$ contains at least one blue point and any blue point is contained in at most $d$ hyperplanes. Therefore, $|\calH| \leq d.d! k_r^d$. By Reduction Rule~\ref{Reduction rule 4}, a hyperplane contains at most $k_r$ red points. Thus $|R| \leq d.d!.k_r^{d+1}$.
	

\end{proof}

%% file: Para_blue.tex

\vspace{0.1cm}
Since every blue point is contained in at most $d$ hyperplanes, the next result is easy to see by the idea of bounded search trees.
\begin{lemma}
	\label{FPT_b} 
	The \AH problem parameterized by $b$ can be solved in time  $O^*(d^b.n^{O(1)})$.
	\end{lemma}
		
%
%

Now we consider the existence of polynomial kernels for this problem.

\begin{thm}

	\label{b_kernal}
	The \AH problem parameterized by $b$ admits a polynomial kernel.
	
\end{thm}
	
\begin{proof}
	
	Assume, we cannot apply reduction rules anymore on $(R \cup B, \mathcal{H},k_r)$, every hyperplane contains at least one blue point. One blue point is covered by at most $d$-hyperplanes. Then the number of hyperplanes $\mathcal{H}$ is at most $db$. To bound the number of red points we reduce an instance of the \AH problem to an instance of the \textsc{Weighted} \AH problem as follows:
		
		 The family of hyperplanes and the set of blue points remain the same as in the original instance. 
		 
		 Let $R_d \subseteq R$ be the set of red points that lie in the intersection of $d$ hyperplanes in $\calH$. For all $r \in R_d$, assign $w(r) =1$. Since any set of $d$ hyperplanes in $\mathbb{R}^d$ intersect at a point, $|R_d| \leq h^d$, where $h$ is the number of hyperplanes. 
	
		 For $i =1 $ to $d-1$, we perform a reduction as follows. For all subsets of exactly $i$ hyperplanes in $\calH$, say,  $H_i =\{h_1, h_2, \dots, h_i\}$, consider the set of red points that is contained in all hyperplanes of $H_i$ and no other hyperplane in $\calH$. In the reduced instance, delete all but one of these red points and assign a weight equal to the number of deleted points $+ 1$. 
		 
		It easy to see that the \AH problem instance has a solution of size $k$ if and only if the \textsc{Weighted} \AH problem instance has a solution of total weight $k$.
		
		In the reduced instance,  the number of red points that lie in the intersection of exactly $i$ hyperplanes is $O(h^i)$. Therefore the total number of red points in the reduced instance is $O(h^d+h^{d-1} +h^{d-2}+\dots+h) = O(h^d) = O((db)^{d})$. 
Now we bound the weight of each red point in the reduced instance. By Reduction rule~\ref{Reduction rule 4}, all hyperplanes in $\calH$ contain at most $k_r$ red points. Therefore, the weight of any red point is at most $k_r$. We see that we need at most $h$ bits to encode the weight of a red point. If not, $k_r > 2^h$. Note that there exists a brute force algorithm that solves the problem in $O^*(2^h)$ time. If $2^h < k_r$, this is a polynomial time algorithm. Hence it follows that the weight can be encoded using $h \leq db$ bits. Hence the size of the reduced instance is $O((db)^{d+1})$.

%
		There exists a  polynomial-time many-one reduction from \textsc{Weighted} \AH problem to \AH problem. Thus, we obtain a polynomial-size kernel for \AH parameterized by $b$.

\end{proof}

%
%

%% file: Para_line.tex
\vspace{.09cm}
We design a parameterized algorithm and a kernel for the \AH problem when parameterized by the number of hyperplanes, $h$. We enumerate all possible subsets of $\mathcal{H}$ and for each subset, we check in polynomial time whether it covers all blue points and at most $k_r$ red points. The algorithm runs in time $O^*(2^{h} (|U | + |F |)$. 

\vspace{0.1cm}
\noindent Now we present a polynomial size kernel to bound the number of blue points we use the following reduction rules

\noindent To bound the number of blue points, we use the following reduction rule. 
\\We apply the following set of reduction rules for all $\delta = 1$ to $ d-1$ exhaustively. For $2\leq \delta \leq d-1$, we apply the reduction rule corresponding to $\delta$ only when the reduction rule corresponding to  $(\delta-1)$ is not applicable. Let $A_{\delta} =  \displaystyle \sum_{i =0}^{\delta} h^i$.\\


\begin{reductionrule}[$\delta$]
\label{PlaneKernel}
Repeat for all $D' \subset [d]$ such that $|D'| = d-\delta$. 
\\Let $B' \subseteq B$ be the set of blue points such that for all $p,q \in B'$, $x_i(p) = x_i(q)$ for all $i \in D'$. If $|B'| > A_{\delta}$ then delete all but $A_{\delta}+1$ points of $B'$.
\end{reductionrule}
\begin{lemma}
Reduction Rule \ref{PlaneKernel} ($\delta$) is safe.
\end{lemma}
\begin{proof}
The proof follows from the fact that in any solution $\mathcal{H'}$, all the points of $B'$ are covered by the same hyperplane. Assume, for contradiction, that no hyperplane in $\mathcal{H'}$ contains all points in $B'$. Thus $B'$ is covered by a set of at most $h$ hyperplanes, say $\mathcal{H}_1 \subseteq \mathcal{H'} $ such that no  hyperplane in $\mathcal{H}_1$ is parallel to the $i^{th}$ co-ordinate axis, for all $i \in D'$. Without loss of generality, assume that $D' = \{1,2,\dots,d-\delta\}$. 

Let $H$ be a hyperplane in $\mathcal{H}_1$ that is parallel to the $j^{th}$ co-ordinate axis for $j \notin D'$. Then for any two points $p,q \in H \cap B'$, $x_j(p)=x_j(q)$ and $x_i(p) =x_i(q)$ for $1\leq i \leq d-\delta$.  Therefore any hyperplane in $\mathcal{H}_1$ can contain at most $A_{\delta-1}+1$  points from $B'$.  Otherwise, Reduction Rule \ref{PlaneKernel} $(\delta-1)$ can be applied, which is a contradiction. 

Thus $\mathcal{H}_1$ covers at most $A_{\delta}$ points. This contradicts our assumption that $\mathcal{H}_1$ covers all points of $B'$. 

\end{proof}	


\begin{thm}
The \AH problem parameterized by $h$ admits a polynomial kernel.
\end{thm}

\begin{proof}	

Let $(R \cup B, \mathcal{H},k_r)$ be an instance of \AH such that none of the reduction rules are applicable. Then, if it is a YES instance, all the points in $B$ can be covered by at most $h$ hyperplanes. Any hyperplane can contain at most $A_{d-1} = O(h^{d-1})$ blue points, otherwise Reduction Rule \ref{PlaneKernel} ($(d-1))$ can be applied. Thus a YES instance can have at most $h^d$ blue points. 

Now that we have bounded $b$ as a polynomial function of $h$, the result follows from Theorem~\ref{b_kernal}.

\end{proof}




%% file: reduction.tex
Given $P=R \cup B$ in the plane and $\mathcal{K}$, a set of $m$ bidirectional skylines, 
 we prove that the \AS problem parameterized by $k_r$ is W[1]-hard by giving a parameterized reduction from the \SQ problem.

\noindent{\RS}: Given a set $\mathcal{S}$ of n rectangles and a set of axis-parallel lines $\mathcal{L}$, the \RS  problem is to decide whether there exists a subset $\mathcal{L}^\prime{} \subseteq  \mathcal{L}$ with $ |\mathcal{L}^\prime{}| \leq k$ such that every rectangle from $\mathcal{S}$ is intersected (stabbed) by at least one line in $\mathcal{L}^\prime{}$. If all the rectangles are unit squares,  it is the \SQ problem. \SQ problem parameterized by $k$ is known to be W[1]-hard~\cite{Dom}.

\begin{figure}

\includegraphics[width=8.5cm]{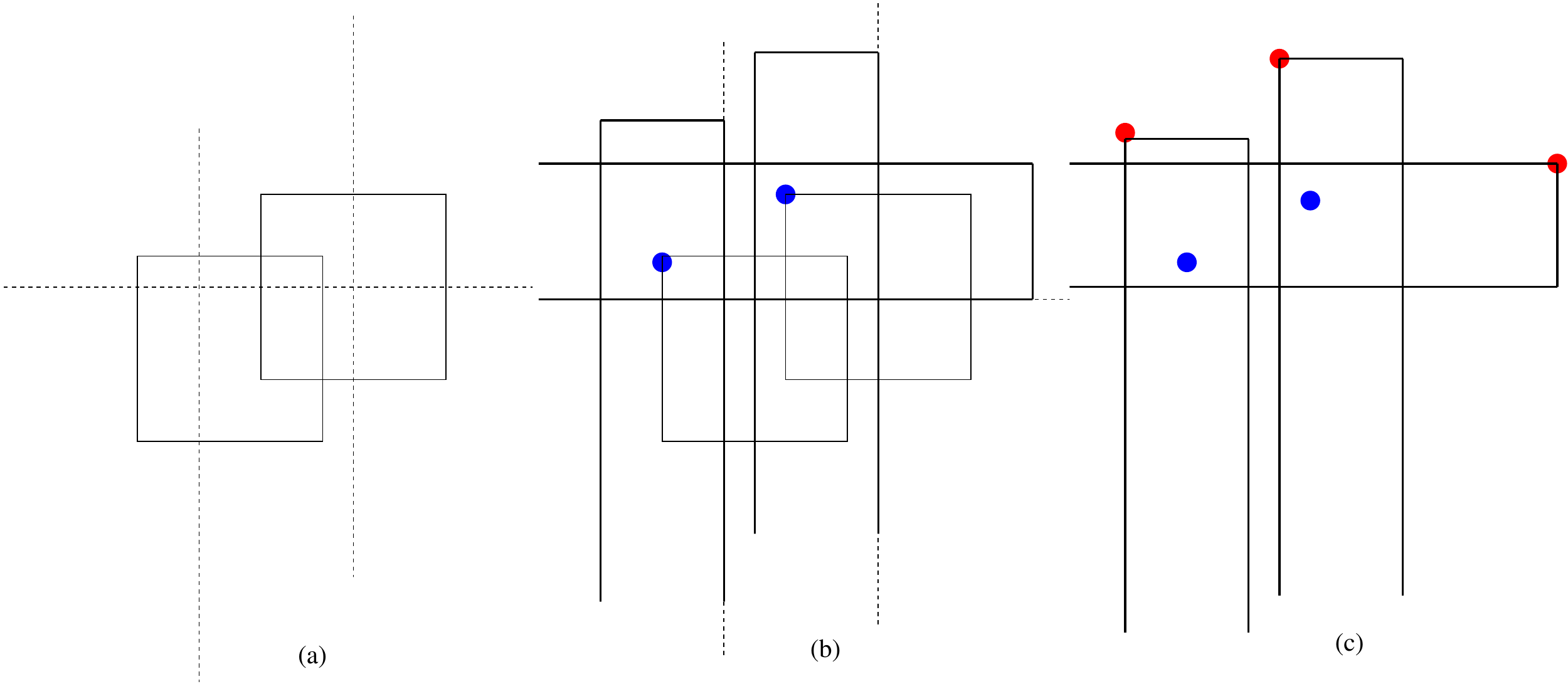}
\centering
\caption{(a)Instance of \SQ problem, (b)Reduction to an instance of \AS, (c)Instance of \AS}
\label{skyline}
\end{figure}

Let ($\mathcal{S},\mathcal{L}$) be an instance of the \SQ problem. Let \{$l_1, l_2,.. l_n$\} be the set of vertical lines in $\mathcal{L}$ which are in non-decreasing order of their $x$-coordinates and let \{$h_1, h_2,.. h_n$\} be the set of horizontal lines in $\mathcal{L}$ which are in non-increasing order of their $y$-coordinates. For every square $S \in \mathcal{S}$, add a blue point $b_S$ to $P$ at the top-left
corner of $S$. For every vertical line $l_i$, we add 
a unit width vertical skyline $S_{l_i}$ with its right edge coinciding with the given line $l_i$. The top edge of the leftmost vertical skyline $S_{l_1}$ is unit distance above
the topmost point in $P$. For $i >1$, the top edge of the skyline $S_{l_i}$ is unit vertical distance away from the top edge of $S_{l_{i-1}}$. Similarly for all horizontal lines, we add a unit width horizontal skyline $S_{h_i}$ whose bottom edge coincides with the line $h_i$. Right edge of the topmost horizontal skyline $S_{h_1}$ is a unit distance away from the rightmost point in $P$. Right edge of $S_{h_i}$ is unit horizontal distance away from the right edge of  $S_{h_{i-1}}$, for $i>1$. For every horizontal skyline, add
a red point $r_{h_i}$ to $P$ in the top-right corner of the skyline and for every vertical skyline, add
a red point $r_{l_i}$ to $P$ in the top-left corner of the skyline. Thus every skyline has a
unique red point. Refer Figure~\ref{skyline}.

\begin{lemma}
 $(\mathcal{S},\mathcal{L})$ can be stabbed using $k$ lines if and only if $(P,\mathcal{K})$ has a solution of size $k$.
\end{lemma}
\begin{proof}
If a line $l\in \mathcal{L}$ stabs a square $S\in \mathcal{S}$, the corresponding skyline $S_l$ covers the blue point $b_S$ and the red point $r_{l_i}$. If a line $l\in \mathcal{L}$ stabs a square $S\in \mathcal{S}$, the corresponding skyline $S_h$ covers the blue point $b_S$ and red point $r_{h_i}$. Now $S$ can be stabbed using $k$ lines if and only if all blue points can be covered using $k$ skylines. Since
every skyline contains a unique red point, $S$ can be stabbed using $k$ lines if and only if $(P,\mathcal{K})$ has a solution of size $k$. Thus
we have proved the following theorem. 
\end{proof}

\begin{thm}
\label{RBSC}

The \RBSC problem on bidirectional skylines is W[1]-hard parameterized by
the solution size.

\end{thm}

Since a family of bidirectional skylines is a subset of the family of axis parallel rectangles, we have the following result.
\begin{corollary}
The \RBSC problem on axis parallel rectangles is W[1]-hard parameterized by
the solution size $k$.
\end{corollary}

%% file: Q1.tex
In this section, we give a polynomial time algorithm to optimally solve the \AQ problem. Given $P= R \cup B$ and a family $\mathcal{Q}_1$ of quadrants of the first type with $|\mathcal{Q}_1| = m$, find a subfamily $\mathcal{Q}'_1$ that covers all points from $B$ and minimum possible number of red points from $R$.

\begin{observation}
	\label{Qb}
	
	Let $q_i$ and $q_j$ be two quadrants such that $(q_i \cap P) \subseteq (q_j \cap P)$ i.e., $(q_i \cap B) \subseteq (q_j \cap B)$ and $(q_i \cap R) \subseteq (q_j \cap R)$ . Then there exist a optimal solution of \AQ , that does not contain both $q_1$ and $q_2$.
\end{observation}

\begin{thm}
	For a given bichromatic point set with $n$ points, \AQ can be computed in $O(mn^2)$ time. 
\end{thm}

\begin{proof}
	Let $\mathcal{O}$ be the orthogonal convexhull of $B$. Consider the left-bottom chain, $P_{lb}$, of $\mathcal{O}$ i.e., the chain of $\mathcal{O}$ that extends from the blue point with the minimum $x$-coordinate to the blue point with the minimum $y$-coordinate, in the anti-clockwise direction.  
	Let $B'=\{b_1, b_2, \cdots b_l\}$  be the set of blue points on $P_{lb}$ ordered by increasing $x$-coordinate. Note that this is also a non increasing order of $y$-coordinate. Let $R' \subseteq R$ and $B'' \subseteq B\setminus B'$ respectively be the set of red and blue points that lie on the right side of $P_{lb}$ (See figure \ref{Q1}). Any set of quadrants that cover all points in $B'$ will cover all points in $B''$ and $R'$. 
	Therefore it is enough to consider the \AQ problem for $(R=(R\backslash R') \cup (B=B'), \mathcal{Q}_1)$. 
	

	For this reduced instance $(R \cup B, \mathcal{Q}_1)$, we introduce the following notations. Let $\pi: R \cup B \longmapsto [n]$ be the bijection corresponding to the ordering of the points in $R \cup B$ by increasing $x$ coordinates and $\sigma: R \cup B \longmapsto [n]$ be the bijection corresponding to the ordering of the points in $R \cup B$ by decreasing $y$ coordinates. For a quadrant $Q \in \mathcal{Q}_1$, let $p \in Q\cap(R\cup B) $ be the point with smallest $x$ coordinate in $Q$ and 
	let $q \in Q\cap(R\cup B) $ be the point with the smallest $y$ coordinate in $Q$. Define $left(Q) =\pi[p]$ and $bottom(Q)=\sigma(q)$. Let $B_i = \{b \in B| \sigma(b)>i\}$, $R_i = \{ r \in R| \sigma(r)>i \}$ and $\mathcal{Q}_j =\{Q\in \mathcal{Q}_1| left(Q)>j\}$.

	
	
	Now we give a dynamic programming algorithm to solve the \AQ problem for this instance.
	
	
	Let $dp[i, j]$ return the minimum number of red points from $R_i$, that is covered by a subset of $\mathcal{Q}_j$ that covers all points in $B_i$.  $dp[0, 0]$ return the value $0$. 

	

	\begin{figure}[tbh]
		\centering
		\includegraphics [width=2\linewidth]{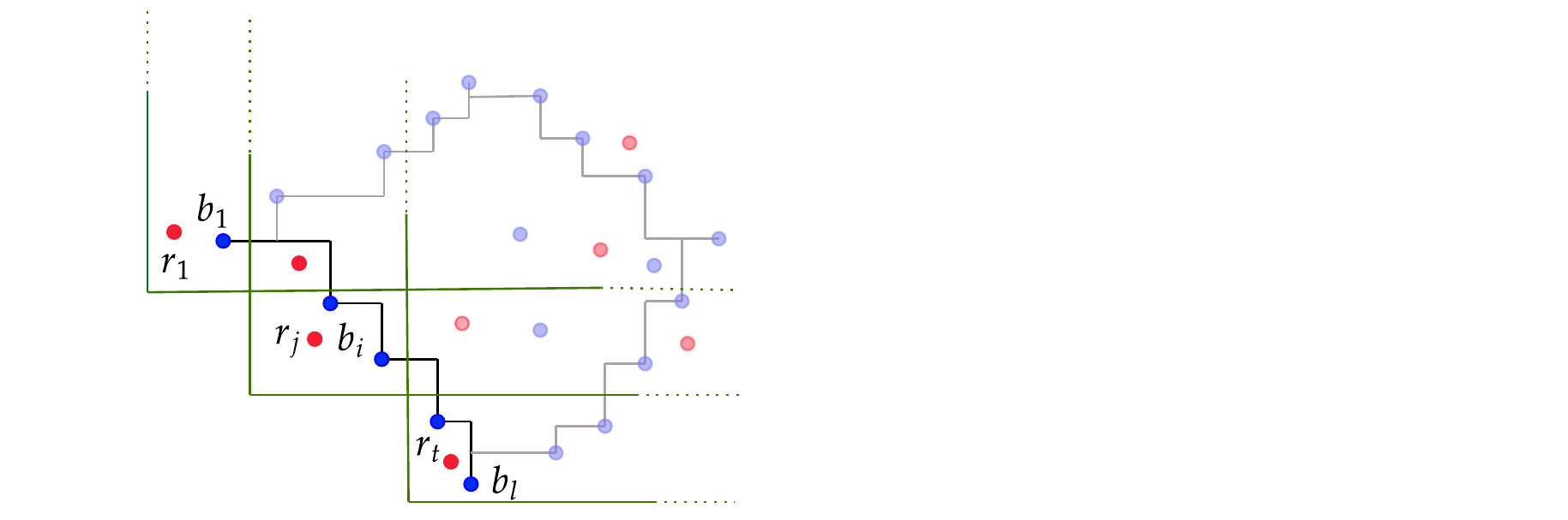}
		\caption{Green line denotes quadrants of the first type and highlighted black orthogonal path indicates left-bottom orthogonal path}
		\label{Q1}
	\end{figure}
	\begin{equation}
	dp[i,j]= \min_{Q\in \mathcal{Q}_{j}, b'\in Q} \big\{ dp[bottom(Q), left(Q)] +|Q \cap R_{i}| \}
	\end{equation}
	
\noindent Here $b'$ is the blue point with the smallest value of  $\sigma(b')$ in $B_i$.
	
	To see the correctness of the above, observe that the recurrence considers all $Q \in \mathcal{Q}_j$ to cover $b'$. For the correct guess, the rest of the solution only covers uncovered red points from $R_{bottom(Q)}$ using $Q \in \mathcal{Q}_{left(Q)}$ . Otherwise, if some $Q \notin \mathcal{Q}_{left(Q)}$ is part of the solution then it contradicts Observation~\ref{Qb}. Similarly an uncovered red point outside $R_{bottom(Q)}$ is not part of an optimal solution. 
	
%
	It is clear that computing one entry in the DP table takes $O(m)$ time. The number of $dp[ *, *]$ values to be computed is clearly $O(n^2)$. We can construct and verify the orthogonal path in $O(n \log n)$ time. 
	
Therefore, the algorithm runs in $O(mn^2)$ time.
	
\end{proof}